%% file: main_KDD.tex
\documentclass[sigconf]{acmart}

\AtBeginDocument{%
  }

\setcopyright{acmlicensed}
\copyrightyear{2026}
\acmYear{2026}
\acmDOI{XXXXXXX.XXXXXXX}

\usepackage{amsmath,amsthm,amsbsy,amssymb,amsfonts,mathtools,multirow,comment,booktabs,textcomp,tablefootnote,adjustbox,graphicx,stfloats,subcaption,makecell}

\usepackage{pgfplots}
\pgfplotsset{compat=1.18}
\usepgfplotslibrary{fillbetween}

\newtheorem{definition}{Definition}
\newtheorem{theorem}{Theorem}
\newtheorem{property}{Property}
\newtheorem{lemma}{Lemma}
\newtheorem{remark}{Remark}
\newtheorem{example}{Example}
\newtheorem{corollary}{Corollary}

\newcommand{\LOO}{\mathtt{LOO}}
\newcommand{\IOI}{\mathtt{IOI}}

\newcommand{\COS}{\mathtt{COS}}
\newcommand{\tA}{\text{A}}
\newcommand{\tB}{\text{B}}
\newcommand{\tC}{\text{C}}

\newcommand{\EE}{\mathtt{EE}}

\newcommand{\FP}{\mathtt{FP}}
\newcommand{\SV}{\mathtt{SV}}
\newcommand{\MRSV}{\mathtt{MR}-\mathtt{SV}}

\definecolor{colA}{HTML}{70d6ff}
\definecolor{colB}{HTML}{ff70a6}
\definecolor{colC}{HTML}{ff9770}
\definecolor{colD}{HTML}{ffd670}

\newcommand{\manualbox}[6]{%
  \addplot[fill=cyan!25, draw=black, thick] coordinates {
    (#1-0.25,#3) (#1+0.25,#3) (#1+0.25,#5) (#1-0.25,#5) (#1-0.25,#3)
  };
  \addplot[draw=orange, thick] coordinates {(#1-0.25,#4) (#1+0.25,#4)};
  \addplot[draw=black, thick] coordinates {(#1,#2) (#1,#3)};
  \addplot[draw=black, thick] coordinates {(#1,#5) (#1,#6)};
  \addplot[draw=black, thick] coordinates {(#1-0.15,#2) (#1+0.15,#2)};
  \addplot[draw=black, thick] coordinates {(#1-0.15,#6) (#1+0.15,#6)};
}

\begin{document}

\title{Beyond \texorpdfstring{\textit{Leave-One-Out}}{Leave-One-Out}: Private and Robust \\ Contribution Evaluation in Federated Learning}
%\title{Beyond \textit{Leave-One-Out}: Private and Robust \\ Contribution Evaluation in Federated Learning}

 \author{Delio Jaramillo Velez}
 \affiliation{%
   \institution{University of La Laguna}
   \city{Tenerife}
   \country{Spain}}
 \email{djaramil@ull.edu.es}

 \author{Gergely Biczok}
 \affiliation{%
   \institution{HUN-REN Hungarian Research Network}
   \city{Budapest}
   \country{Hungary}}
 \email{biczok@crysys.hu}
 
  \author{Alexandre Graell i Amat}
 \affiliation{%
 	\institution{Chalmers University}
 	\city{Gotheburg}
 	\country{Sweden}}
 \email{alexandre.graell@chalmers.se}
 
  \author{Johan Ostman}
 \affiliation{%
 	\institution{AI Sweden}
 	\city{Gotheborg}
 	\country{Sweden}}
 \email{johan.ostman@ai.se}

 \author{Balazs Pejo}
 \affiliation{%
   \institution{Budapest University of Technology and Economics}
   \city{Budapest}
   \country{Hungary}}
 \email{pejo@crysys.hu}

\begin{abstract}
Cross-silo federated learning allows multiple organizations to collaboratively train machine learning models without sharing raw data, but client updates can still leak sensitive information through inference attacks. Secure aggregation protects privacy by hiding individual updates, yet it complicates contribution evaluation, which is critical for fair rewards and detecting low-quality or malicious participants. Existing marginal-contribution methods, such as the Shapley value, are incompatible with secure aggregation, and practical alternatives, such as Leave-One-Out, are crude and rely on self-evaluation.

We introduce two marginal-difference contribution scores compatible with secure aggregation. Fair-Private satisfies standard fairness axioms, while Everybody-Else eliminates self-evaluation and provides resistance to manipulation, addressing a largely overlooked vulnerability. We provide theoretical guarantees for fairness, privacy, robustness, and computational efficiency, and evaluate our methods on multiple medical image datasets and CIFAR10 in cross-silo settings. Our scores consistently outperform existing baselines, better approximate Shapley-induced client rankings, and improve downstream model performance as well as misbehavior detection.
These results demonstrate that fairness, privacy, robustness, and practical utility can be achieved jointly in federated contribution evaluation, offering a principled solution for real-world cross-silo deployments.

\end{abstract}

% CCS Concepts



%\received{1 February 2026}
%\received[revised]{1 May 2026}
%\received[accepted]{1 June 2026}

\maketitle

\section{Introduction}
\label{sec:intro}

Federated learning (FL)~\cite{kairouz2021advances} is a distributed machine learning approach in which multiple parties jointly train a machine learning model without exchanging their raw data. Instead of centralizing the training data, FL delegates computation to the clients and communicates only model updates (e.g., gradients) with a central server. This decentralized approach enhances privacy by design, as sensitive data never leaves the clients.

In addition to privacy preservation, contribution evaluation (CE)~\cite{siomos2023contribution} is a critical component of FL. Fairly assessing client contributions is essential for incentivizing honest participation, ensuring equitable reward distribution, and detecting under-performing or malicious actors. CE schemes are specifically designed to assign a \emph{score} to each participant that reflects their importance toward the collaborative process.

While FL provides built-in privacy protection by keeping raw data local, individual model updates can still leak sensitive information through inference attacks~\cite{Leak}. To mitigate this risk, additional privacy-preserving mechanisms such as differential privacy and secure aggregation have been introduced. In particular, secure aggregation (SA)~\cite{zhou2022survey} has emerged as a standard technique for concealing individual client updates, making it harder to perform inference attacks. Based on lightweight secure multiparty computation~\cite{cramer2015secure} and secret sharing~\cite{Shamir79}, secure aggregation ensures that individual client updates remain hidden from the server while still enabling correct computation of the aggregated global model.
SA therefore poses a fundamental tension with CE, as most prior methods depend on individual client updates~\cite{liu2021gtg,ijcai2022p782} and are incompatible with the SA mechanism.

One of the most widely studied CE methods is the Shapley value ($\SV$)~\cite{shapley1953value}, a marginal-contribution-based metric rooted in cooperative game theory. It is uniquely characterized by a set of well-established fairness properties, making it a principled foundation for fair reward allocation.  However, despite its theoretical appeal, the $\SV$ is computationally intractable for more than a few clients, as it requires training exponentially many models. Moreover, its reliance on marginal contributions makes it fundamentally incompatible with privacy-enhancing mechanisms, such as SA. Well-known approximations inherit this incompatibility. 

The Leave-One-Out mechanism ($\LOO$)~\cite{black2021leave,evgeniou2004leave} remains the only practical marginal difference based CE method compatible with SA. However, this approach is severely limited, yielding a crude approximation of the $\SV$. Furthermore, it satisfies only a subset of the desirable fairness properties and requires self evaluation when SA is applied, making it unreliable in the presence of selfish clients who seek to manipulate their own scores.

\paragraph{Our contribution.}
In this work, we address the fundamental problem of designing CE methods for cross-silo FL that simultaneously achieve fairness (by satisfying the desired properties),  privacy (through compatibility with SA), and robustness against selfish clients. Our contributions are the following.

\begin{itemize}
    \item We propose two novel CE scores that are compatible with SA: Fair-Private ($\FP$) and Everybody-Else ($\EE$). $\FP$ satisfies all fairness properties, but still relies on self evaluation. On the contrary, $\EE$ deliberately relaxes the null player property to provide robustness against selfish clients.
      
    \item We establish theoretical guarantees for the proposed CE methods, covering fairness, privacy, computational complexity, and resistance to manipulation.
              
    \item We empirically evaluate the proposed scores on three medical image datasets and CIFAR10. Our methods consistently outperform $\LOO$ and better approximate the client contributions induced by the ground truth $\SV$.    
\end{itemize}

\section{Related work} 
\label{sec:rw}

The privacy-utility trade-off in CE is mostly studied in the literature from two angles. The first examines how existing CE methods amplify privacy leakage during learning, while the second focuses on designing CE schemes that yield meaningful and reliable scores while preserving privacy. The former has received considerably more attention. A recent survey~\cite{lin2024comprehensive} highlights that most existing work at the intersection of CE and privacy focuses on how different CE metrics may increase the risk of inference attacks. In contrast, the latter problem, how to construct CE methods that are intrinsically privacy-preserving, remains mostly neglected.

Only a few works have explicitly addressed privacy-preserving CE. The authors of~\cite{zheng2022secure} proposed a multi-server solution based on encryption, which incurs significant computational overhead. Another paper~\cite{watson2022differentially} used Differential Privacy by injecting noise to client updates before evaluation, which severely degrades utility. Finally, in~\cite{ma2021transparent}, the researchers proposed a blockchain-based framework that requires a drastic architectural change. Critically, none of these approaches considers SA, despite its widespread adoption in practical FL systems. 

To the best of our knowledge, only three prior works have addressed CE under SA: Federated Group Testing (FedGT)~\cite{xhemrishi2023fedgt}, Quality Inference (QI)~\cite{pejo2023quality}, and their combination~\cite{xhemrishi2025detect}. However, these operate under different threat models. FedGT relies on an assignment matrix that exposes specific sub-coalitions, thereby reducing the anonymity set of the SA protocol. QI relies on the probabilistic randomness of client selection in partial-participation (Cross-Device) settings. In contrast, our work addresses strict Cross-Silo FL with full participation and maximal anonymity, where these methods are not directly applicable. Moreover, unlike our work, these methods lack built-in fairness properties.

\section{Preliminaries}
\label{sec:pre}

\paragraph{Notation.} Our notation is summarized in Table~\ref{tab:not}. We denote by $M_0$ and $M$ the global model before and after a training round, respectively, and by $M_i$  the local model of client $i\in[N]=\{1,\dots,N\}$. The local update of client $i$ is defined as the pseudogradient $U_i=\frac{1}{N}(M_i-M_0)$. We assume that aggregation is performed by simple addition, i.e.,  $M=M_0+\sum_{i=1}^NU_i$. This setting encompasses the standard aggregation technique \emph{Federated Averaging (FedAvg)}~\cite{kairouz2021advances}. We assume SA~\cite{bonawitz2017practical} is employed, i.e., the individual client updates ($U_i$) are hidden from the server. Yet, it can still compute the aggregated global model ($M$), preserving both model functionality and client privacy.

\begin{table}[!t]
 \caption{Notation}
  \label{tab:not}
    \centering
    \renewcommand{\arraystretch}{1.1}
    \begin{tabular}{c|l}
        Symbol & Meaning \\
        \hline
        $N$ & Number of clients \\
        $M_0$ & Initial model before the training round \\
        $M$ & Aggregated model after the training round  \\
        $U_i$ & Client $i$ local update within the round \\
        $v(\cdot)$ & Evaluation function (loss, accuracy, etc.)  \\
    \end{tabular}
\end{table}

\subsection{Contribution Evaluation}
\label{sec:CE}

To establish an objective reference point for contribution assessment, most CE schemes employ an external dataset to evaluate the aggregated models. This approach is based on the premise that the value of a dataset correlates with the performance of a model trained on it, thereby creating an intrinsic coupling between data valuation and validation~\cite{song2019profit,liu2021gtg}. Other methods estimate contributions by measuring similarity between the global model and individual client models, e.g., using Cosine Similarity ($\COS$)~\cite{xu2021gradient}. However, this approach suffers from fundamental limitations: $\COS$ captures only directional alignment while ignoring differences in magnitude, feature importance, and decision boundaries~\cite{Cosine_pitfalls}. Consequently, two models may appear highly aligned under this metric while exhibiting drastically different behaviors and generalization performance. Moreover, using the aggregated global model as the baseline introduces additional bias, as it may incorporate malicious or low-quality client updates, rendering it an unreliable anchor~\cite{collabotative}.

In this work, in line with the marginal difference based CE literature~\cite{ghorbani2019data, song2019profit}, we assume the availability of an external representative test dataset that reflects the overall population. Such a dataset could be constructed by a Generative Adversarial Network (GAN) generating a representative synthetic local test set~\cite{con_gan}. What is more, such a process can be infused with Differential Privacy~\cite{pejo2022guide}, making it a suitable approach for our privacy-preserving setting~\cite{privacy_gan_one,privacy_gan_two}. For instance, the synthetic local test sets can then be sent to the server to be combined to create a global test set~\cite{gan_sv}.

\paragraph{Shapley value.}
In FL, the most widely used CE methods relying on a test set are approximations of the \textit{Shapley value} ($\SV$)~\cite{rozemberczki2022shapley}, a concept rooted in cooperative game theory. The $\SV$ of a client is defined as the weighted sum of its marginal contributions to all possible coalitions $S\subseteq[N]$. It relies on the utility function $v(\cdot)$, which assigns a value to each coalition of clients. In FL, $v$ typically measures model performance (e.g., accuracy or loss).

\begin{definition}[Shapley value~\cite{shapley1951notes}]
    \label{def:shap}
    The Shapley value of client $i$ for a given utility function $v$ is formalized below.
    \begin{equation}
        \label{eq:shap}
        \SV(i)= \frac1N\sum_{S \subseteq [N] \setminus \{i\}} \binom{N-1}{|S|}\cdot(v(S \cup \{i\}) - v(S))
    \end{equation}
\end{definition}

The $\SV$ is the only mechanism satisfying four fundamental properties (\textit{linearity}, \textit{efficiency}, \textit{null player}, and \textit{symmetry}), widely accepted as collectively characterizing fair reward allocation. The linearity property assumes contributions are additive and independent, an assumption that generally does not hold in FL due to the nonlinear and sequential nature of the client updates. Therefore, linearity has limited practical relevance in FL~\cite{song2019profit}. As such, the primary design goal for a fair CE method is to satisfy the following three properties.

\begin{property}[Efficiency]
    \label{ax:eff}
    The total value generated by the grand coalition must be fully distributed among all players.
\end{property}

\begin{property}[Null player]
    \label{ax:null}
   If a player does not contribute any additional value to any coalition, its score should be zero.
\end{property}

\begin{property}[Symmetry]
    \label{ax:sym}
    If two players contribute equally to every possible coalition, they should receive the same score. 
\end{property}

Exact computation of the $\SV$ requires training separate models for every possible subset of clients $S \subseteq [N]$, a process referred to as the \textit{retraining game}. This approach is computationally infeasible, as it requires training an exponential number of models. To address this limitation, an approximation method has been proposed~\cite{song2019profit} that leverages per-round gradient updates during training a single global model. In each round, the method evaluates all possible subsets of gradient updates to compute client scores, which are then averaged across rounds.  We refer to this approximation as the multi-round Shapley value ($\MRSV$). 

A major drawback of these approaches is their incompatibility with SA, as they require access to individual client updates. Under SA, only a limited set of submodels is available for evaluation at each round. The server has access only to $M_0$ and $M$, while client $i$ can compute two additional coalitions: $M_0+U_i$ and $M-U_i$. Consequently,  only clients can compute marginal difference based CE scores using this restricted set of submodels. One approach to approximate the $\SV$ based on these submodels is \textit{Leave-One-Out} ($\LOO$)~\cite{black2021leave,evgeniou2004leave}, which evaluates a client's contribution by comparing the performance of the model of the grand coalition $S=[N]$ against that of the model where the local update of the client is omitted, $S'=[N]\setminus\{i\}$. 

\begin{definition}[Leave-One-Out]
    \label{def:L1O}
    The Leave-One-Out score of client $i$ is defined as $\LOO(i)=v(M) - v(M-U_i)$.
\end{definition}

Importantly, only client $i$ can evaluate both $v(M_0+U_i)$ and $v(M-U_i)$, since its update $U_i$ is hidden from other clients and the server. This requires clients to have access to the global test dataset, which is shared by the server, to perform these evaluations. More problematically, it turns $\LOO$ into a \textit{self evaluation} technique where each client computes and reports its own score to the server, thus becoming sensitive to manipulation by selfish clients. 

\subsection{Threat Model and Goal}
\label{sec:threatmodel}

\begin{table*}[!b]
    \caption{Threat model summary}
       \centering
    \label{tab:threatmodel}
    \begin{tabular}{p{0.1\textwidth}|p{0.27\textwidth}|p{0.3\textwidth}|p{0.2\textwidth}}
        \textbf{Entity} & \textbf{Assumptions} & \textbf{Capabilities} & \textbf{Goals} \\
        \hline\hline
        Server & Honest-but-curious; follows protocol; not malicious & Receives masked local updates; aggregates local update; aggregates local scores  & Coordinate training; learn aggregated model \\
        \hline
        Client (training) & Honest-but-curious; follows training protocol; uses true local data & Holds local dataset; sends masked local update to server & Achieve best possible global model \\
        \hline
        Client (post training) & Selfish-but-non-malicious; does not disrupt training & Knows local and aggregated update; estimates  marginal differences & Maximize own contribution score or minimize others’ \\
    \end{tabular}
\end{table*}

We study cross-silo federated learning with secure aggregation. Strict compatibility is essential for preventing leakage of sensitive training data via white-box attacks such as Deep Leakage from Gradients (DLG)~\cite{zhu2019deep}. In regulated environments (e.g., GDPR), exposing individual gradients is often legally prohibited. While some attacks target aggregated gradients, SA ensures that any potential leakage cannot be traced to specific participants, making it a critical Privacy Enhancing Technology (PET) for cross-silo consortia~\cite{heyndrickx2023melloddy}.

Our trust assumptions are summarized in Table~\ref{tab:threatmodel}. The server is \emph{honest-but-curious}: it executes the training protocol faithfully without malicious behavior, but may attempt to learn additional information from aggregated values. All clients are assumed to be \emph{selfish}, prioritizing their own utility and, thus, reward maximization. For utility maximization, the best known strategy is honest training on local data to collectively obtain the best possible global model. Therefore, clients adhere to the prescribed update protocol. However, contribution evaluation, which takes place \emph{after} each training round as a post-processing step, introduces different incentives. Clients are motivated to maximize their contribution scores either directly through inflated self reporting or indirectly by undermining others' scores.  We exclude malicious behavior such as Sybil attacks, treating clients as independent, selfish actors without collusion capabilities. These assumptions are in line with other works, such as~\cite{anada2025measuring}. Given this, we aim to design CE mechanisms that preserve privacy by ensuring compatibility with SA and prevent manipulation of contribution reports.

\section{Private Contribution Evaluation}
\label{sec:mod}

We focus on contribution scoring in a cross-silo setting, specifically evaluating contributions within a single training round, similar to~\cite{song2019profit}.
 Under SA, the only models accessible to all clients are the grand coalition $S=[N]$ and the empty coalition $S=\emptyset$, i.e., $M$ and $M_0$, respectively.  However, client $i$ individually has access to the model corresponding to the singleton coalition $S=\{i\}$, given by $M_0+U_i$. Moreover, by removing its own update from the global model, client $i$ can reconstruct the model corresponding to $S=[N]\setminus\{i\}$, i.e., $M-U_i$.

Building on the above, we formalize the additional privacy property that a CE method must have to be compatible with SA, alongside the three desirable fairness properties.

\begin{property}[SA Compatibility]
    \label{ax:priv}
    A CE scheme is compatible with SA if the score assigned to each client depends only on the models corresponding to the coalitions $[N]$, $\emptyset$, $[N]\setminus\{i\}$, and $\{i\}$ $\forall i\in [N]$.
\end{property}

\subsection{Self Evaluation}
\label{sec:FP}

In this section, we introduce a novel CE score, termed \emph{Fair-Private} ($\FP$). The key idea is to leverage the singleton coalition $\{i\}$ available to client $i$ to construct a CE score more expressive than the standard $\LOO$ score.

\begin{definition}[Fair-Private]
    \label{def:FP2}
    The Fair-Private score of client $i$ is defined in~\eqref{eq:FP} where $\alpha$ is defined in~\eqref{eq:alphai}. 
    
    \begin{equation}
        \label{eq:FP}
        \FP(i)=\left(\cfrac{ \alpha(i)}{\sum_{j=1}^{N}{ \alpha(j)}}\right)\cdot v(M)
    \end{equation}

    \begin{equation}
    \label{eq:alphai}
    \alpha(i)=\frac{[v(M) - v(M-U_i)]+[v(M_0+U_i)-v(M_0)]}{2}
    \end{equation}
\end{definition} 

The first term in~\eqref{eq:alphai} is the Leave-One-Out score, while the second is its symmetric counterpart `Include-One-In' ($\IOI$): instead of removing $i$ from the grand coalition, we add it to the empty set. The combination of these marginal differences mirrors the $\SV$, where coalitions of complementary sizes receive identical weights.  Our construction of the $\FP$ score further draws inspiration from the normalized Banzhaf power index~\cite{banzhaf}. By applying the same type of normalization in~\eqref{eq:FP}, we can enforce the efficiency property (Property~\ref{ax:eff}) while simultaneously preserving the other fairness and privacy properties, as established in the following theorem. 

\begin{theorem}[Properties of $\FP$]
    \label{th:L1O-I1I_ax}
     $\FP$ satisfies Properties~\ref{ax:eff},~\ref{ax:null},~\ref{ax:sym}, and~\ref{ax:priv}.
\end{theorem}

\begin{proof}
    The proof is given in Appendix~\ref{sec:appTh1}.
\end{proof}

$\FP$ can be formalized within the family of semi-values~\cite{with_effic}. While $\FP$ mathematically resembles an arithmetic mean of the two semi-values ($\LOO$ and $\IOI$), its derivation is strictly constraint-driven. Under SA, the space of computable semi-values effectively collapses due to the limited observable marginals. Thus, the term $\alpha(i)$ (see Equation~\ref{eq:FP}) represents the most expressive semi-value constructible using all available information without violating the privacy constraint. It is also important to note that while $\LOO$ (and $\IOI$) satisfies the Linearity axiom, $\FP$ loses that due to the rescaling. In fact, it is impossible to compute a score that is simultaneously Linear, Efficient, and SA-compatible. With rescaling, we prioritize Efficiency, which is practically indispensable for incentive mechanisms in fixed budget federated environments. We refer the reader to Appendix~\ref{app:semivalues} where we further elaborate on these aspects. 

\subsection{Evaluating Others}
\label{sec:EE}

Note that all previous schemes (e.g., $\LOO$, $\FP$) directly used $U_i$ to assign a score to client $i$. Thus, those scores can only be computed by the respective client. Such self evaluation scores have an inherent drawback: a selfish client can misreport or manipulate the performance of its model, thereby gaining an unfair advantage and undermining the integrity of the CE process. To overcome this limitation, we introduce an alternative CE score termed \emph{Everybody-Else} ($\EE$) that, in addition to being compatible with SA, does not rely on self evaluation, thereby providing resistance to direct manipulation.

\begin{property}[Manipulation Resistance]
    \label{ax:inc}
    A CE method is manipulation-resistant if no client can directly manipulate their own score.
\end{property}

To satisfy Property~\ref{ax:inc}, we introduce the idea of having clients evaluate each other, rather than themselves. This ensures that no client can directly influence its own score.

\begin{definition}[Everybody-Else]\label{def:EE}
   The Everybody-Else score is defined in~\eqref{eq:EE} where $\beta$ and $\gamma$ are defined in~\eqref{eq:LEEO-IEEI}.
   
   \begin{equation}
        \label{eq:EE}
        \EE(i)=\left(\frac{ \frac{\beta(i)+\gamma(i)}{2}}{\sum_{j=1}^{N}{ \frac{\beta(i)+\gamma(i)}{2}}}\right)\cdot v(M)
    \end{equation}
    
    {\footnotesize\begin{equation}
    \label{eq:LEEO-IEEI}
        \beta(i)=\sum\limits_{j\not=i}\frac{[v(M)-v(M_0+U_j)]}{(N-1)^2}
        \hspace{0.1cm};\hspace{0.1cm}
        \gamma(i)=\sum\limits_{j\not=i}\frac{[v(M-U_j)-v(M_0)]}{(N-1)^2}
    \end{equation}}
\end{definition}

Equation~\eqref{eq:EE} contains the same rescaling as before, while~\eqref{eq:LEEO-IEEI} formalizes our idea of having clients evaluate each other. The score of client $i$ is determined collectively by the remaining clients $j\in[N]\setminus\{i\}$. Direct computation of client $i$'s score via marginal difference based evaluation is infeasible, as no other client $j$ can compare two coalitions that differ solely in the presence or absence of client $i$. Instead, our key insight is that the contribution of client $i$ can be inferred indirectly by assessing the collective contribution of the other clients. Specifically, each client $j\not=i$  computes marginal differences involving itself. $\beta(i)$ measures the effect of comparing the grand coalition $[N]$ against $\{j\}$, while $\gamma(i)$ measures the effect of comparing $[N]\setminus\{j\}$ against $\emptyset$. Aggregating these evaluations across all $j \neq i$, we obtain an estimate of client $i$'s score without requiring self evaluation.

Across all $j \neq i$, client $i$ is the only participant that appears in every coalition considered in the summations for $\beta(i)$ and $\gamma(i)$, while the other clients appear less frequently. Since client $i$ is systematically present, the aggregation of these marginal differences over $j \neq i$ yields a quantity that predominantly captures the contribution of client $i$. Hence, $\EE(i)$ can be interpreted as an indirect estimate of client $i$'s contribution, even though it is constructed solely from the evaluations of all other clients.

In line with the principles of the $\SV$, each marginal difference in the summations of $\beta(i)$ and $ \gamma(i)$ is weighted equally, since all of them are computed over coalitions of the same size. The normalization by $(N-1)^2$ serves two purposes. First, the summation aggregates contributions from $N-1$ clients (excluding client $i$) and dividing by $N-1$ yields their average. Second, each term corresponds to a coalition of $N-1$ clients (excluding client $j$), and we distribute its value uniformly across the remaining $N-1$ clients. This normalization ensures that the resulting scores are on the same scale as the $\SV$, as confirmed in our experiments.  

For clarity, we next illustrate the definition with an example, followed by our theoretical results. 

\begin{example}[Three clients]
    Consider a scenario with three clients: $A$, $B$, and $C$. The score of client $A$ is determined from the evaluations provided by clients  $B$ and $C$. Specifically, client $B$ computes $\gamma(A)$ as the marginal difference between the coalition $\{A,C\}$ and the empty set $\emptyset$, and computes $\beta(A)$ by comparing $\{B\}$ with the grand coalition $\{A,B,C\}$. In this way, client $B$ effectively evaluates the set $\{A,C\}$. Similarly, client $C$ performs the same type of evaluations to measure the contribution of $\{A,B\}$. Combining these results in $\EE(A)$ produces a score that reflects both $\{A,C\}$ and $\{A,B\}$. In these coalitions, all clients are included at least once, but only client A appears in every coalition. Because client A is systematically present, aggregating the differences across clients B and C yields a quantity that primarily reflects client A's contribution.
    
\end{example}

\begin{theorem}[Properties of $\EE$]
    \label{th:LEEOIEEI}
    $\EE$ satisfies Properties~\ref{ax:eff},~\ref{ax:sym},~\ref{ax:priv}, and~\ref{ax:inc}.
\end{theorem}

\begin{proof}
    The proof is given in Appendix~\ref{sec:appTh2}.
\end{proof}

Note that on the contrary, $\EE$ does not satisfy the null player axiom (Property~\ref{ax:null}), as illustrated in the following example. 

\begin{example}\label{ex:null}
    Consider a scenario with three players where the coalition values are given by  
    $v(\emptyset)=0$, $v(\{\tA\})=0$, $v(\{\tB\})=1$, $v(\{\tC\})=2$, $v(\{\tA,\tB\})=1$, $v(\{\tA,\tC\})=2$, $v(\{\tB,\tC\})=3$, and $v(\{\tA,\tB,\tC\})=3$. 
    Here, client A is clearly a null player, as adding A to any coalition does not increase its value. 
    However, the scores computed by $\EE$ are positive, namely $\EE(\tA)=3/4$, demonstrating that $\EE$ does not satisfy Property~\ref{ax:null}.
\end{example}

In a sense, the EE score sacrifices the Null Player property to achieve manipulation resistance. We argue this is a small trade-off in the context of FL. In cooperative game theory, a `null player' contributes zero value. However, in FL, `free-riders' are strategic actors who can easily inject noise or random weights rather than submitting detectable zero gradients. EE prioritizes robustness (Property 5) over identifying passive contributors, which is more practical for mitigating active selfish behavior in FL environments.

\subsection{Overview}
\label{sec:over}

To facilitate understanding of our proposed schemes, their properties, and their interconnections, we provide a summary in Table~\ref{tab:coalitions}. The table not only indicates which scheme satisfies each axiom but also highlights the underlying building blocks on which each scheme relies. We then analyze the computational complexity of each scoring method.

{\small
 \begin{table}[!b]
  \caption{The building blocks for the scores and the properties they satisfy, e.g., $\LOO(i) = v(M) - v(M \setminus U_i)$, which depends only on $M$ and $U_i$.}
     \label{tab:coalitions}
     \centering
     \begin{tabular}{c|ccc|ccc|ccccc}
         Cont. & \multirow{2}{*}{$M_0$} & \multirow{2}{*}{$+U_i$} & \multirow{2}{*}{$+U_{j}$} & \multirow{2}{*}{$-U_{j}$} & \multirow{2}{*}{$-U_i$} & \multirow{2}{*}{$M$} & \multicolumn{5}{c}{Property} \\
         Eval.&&&&&&&1&2&3&4&5\\
         \hline
          $\LOO(i)$ & & & & & $\medbullet$ & $\medbullet$ &&$\checkmark$&$\checkmark$&$\checkmark$&\\
          $\FP(i)$ & $\medbullet$ & $\medbullet$ & & & $\medbullet$ & $\medbullet$ &$\checkmark$&$\checkmark$&$\checkmark$&$\checkmark$&\\
          $\EE(i)$ & $\medbullet$ &  & $\medbullet$ & $\medbullet$ &  & $\medbullet$ &$\checkmark$&&$\checkmark$&$\checkmark$&$\checkmark$\\ 
     \end{tabular}
 \end{table}}

\begin{theorem}[Complexity]
    \label{th:complex}
    The computational complexity of $\FP$ and $\EE$ is $\mathcal{O}(N)$.
\end{theorem}
\begin{proof}
The proof is given in Appendix~\ref{sec:appTh3}.
\end{proof}

\section{Experiments}
\label{sec:exp}

\paragraph{Setup.} We evaluate our CE schemes in a cross-silo FL scenario using four image datasets: ISIC2019~\cite{ISIC}, PatchChameleon~\cite{Pcam}, Brain-MRI~\cite{brain}, and CIFAR10~\cite{CIFAR10}. We use the ResNet architecture for ISIC2019 and PatchChameleon and a CNN consisting of 4 layers (Conv2d(3,16,3), Conv2d(16, 32, 3), Linear(100352,128), and Linear(128, 4)) for Brain-MRI and CIFAR10 with and without max pooling, respectively. We utilized the Adam optimizer with both learning rate and weight decay of $0.001$. Further details, such as the number of local epochs, are provided in Appendix~\ref{app:Experiments}.

Following prior work, we assume a common test set for CE that contains the same number of samples per class and is available to all clients~\cite{song2019profit,liu2021gtg}. For each dataset, we construct a non-IID baseline by splitting the data across clients according to a Dirichlet distribution with parameter $\mu=0.5$. All results are reported across 10 runs, and we report both mean and variance. 

The estimation of the privacy-preserving scores follows a star topology, in which clients do not communicate directly with each other. Each client computes the marginal differences for its available coalitions and reports them to the server. The server then aggregates these reported differences to obtain the final scores. To support reproducibility, we have open-sourced our implementation~\footnote{\url{https://anonymous.4open.science/r/Fair-privacy-Preserving-Contribution-Scoring--8C5B/README.md}}.

\paragraph{Evaluation metrics.}
 We compute all scores for a single round (10th by default), yet the baseline CE score is the Multi-Round approximation of the Shapley value ($\MRSV$) as described in~\cite[Algorithm 2]{song2019profit}. We also compute the cosine similarity ($\COS$) accumulated over all training rounds.

In CE, the relative ranking of clients is often more relevant, particularly for identifying the most and least valuable contributors. We therefore report the Spearman correlation coefficient $\phi$~\cite{zar2005spearman}, the Kendall rank $\kappa$~\cite{abdi2007kendall}, and the Pearson value $\rho$~\cite{cohen2009pearson}. In our experiments, all three metrics showed consistent results, so we primarily focus on the Spearman coefficient. It measures how closely our schemes' rankings match the ground truth. A value of $\phi=1$ indicates perfect alignment with the baseline ranking, $\phi=-1$ indicates a perfectly inverted ranking, and values near zero suggest no correlation. 

Naturally, the absolute values of the scores are also important. Since CE scores can be arbitrarily scaled and may take negative values, we normalize them in two steps before comparison: (i) subtracting the minimum score to ensure non-negativity, and (ii) dividing by the mean to account for scale differences. In this normalized space, the $L_2$ distance effectively captures the spatial proximity of the scores, which is critical for fair reward allocation. For instance, a score vector of $[0.01, 0.02, 0.97]$ implies drastically different financial or resource incentives compared to $[0.32, 0.33, 0.35]$, even though the two vectors share the same ranking (imagine distributing \$10,000 proportional to those scores).

 \begin{figure}[!t]
    \centering
    \caption{Various CE scores for a scenario with 4 clients ($\mu=0.5$) at the 10th communication round using the ISIC2019 dataset.}
    \label{fig:illustrate}
      \begin{tikzpicture}[scale=0.7] %
\begin{axis}[
    ybar,
    bar width=6pt,
    width=12cm,
    height=9cm,
    ylabel={Scores},
    xlabel={},
    symbolic x coords={$\MRSV$,$\LOO$,$\FP$,$\EE$,$\COS$},
    xtick=data,
    ymin=0,
    ymax=1.7,
    enlarge x limits=0.15,
    tick style={draw=none}, % <--- removes the small dashes
    legend style={
        at={(0.97,0.97)}, % top-right inside
        anchor=north east,
        draw=black,
        fill=white,
        font=\small,
        rounded corners=2pt,
        legend columns=2,
        column sep=8pt
    },
    title={},
    legend image post style={mark=none, draw=none},
    area legend
]

% Hospital A
\addplot+[ybar,fill=colA] coordinates {
    ($\MRSV$,1.42)
    ($\LOO$,0.80)
    ($\FP$,1.19)
    ($\EE$,1.32)
    ($\COS$,0.98)
};

% Hospital B
\addplot+[ybar,fill=colB] coordinates {
    ($\MRSV$,0.80)
    ($\LOO$,1.60)
    ($\FP$,0.95)
    ($\EE$,0.9)
    ($\COS$,1.00)
};

% Hospital C
\addplot+[ybar,fill=colC] coordinates {
    ($\MRSV$,0.60)
    ($\LOO$,0.40)
    ($\FP$,0.80)
    ($\EE$,0.65)
    ($\COS$,1.00)
};

% Hospital D
\addplot+[ybar,fill=colD] coordinates {
    ($\MRSV$,1.18)
    ($\LOO$,1.20)
    ($\FP$,1.05)
    ($\EE$,1.16)
    ($\COS$,1.01)
};

% --- Manual legend with only squares ---
\addlegendimage{only marks, mark=square*, mark options={fill=colA, draw=black}}
\addlegendentry{Hospital A}
\addlegendimage{only marks, mark=square*, mark options={fill=colB, draw=black}}
\addlegendentry{Hospital B}
\addlegendimage{only marks, mark=square*, mark options={fill=colC, draw=black}}
\addlegendentry{Hospital C}
\addlegendimage{only marks, mark=square*, mark options={fill=colD, draw=black}}
\addlegendentry{Hospital D}

\end{axis}
\end{tikzpicture}
\end{figure}
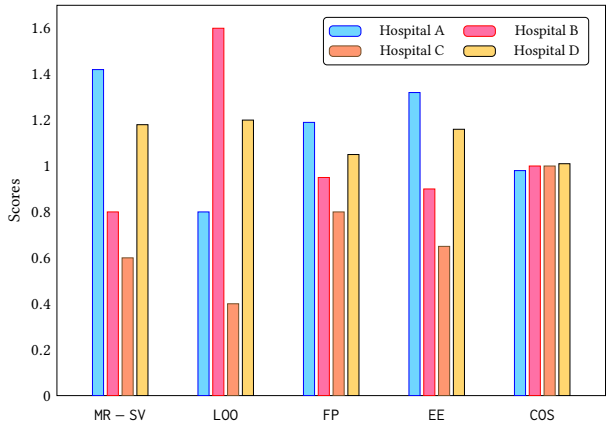

\subsection{Comparative Results}

\paragraph{Comparison to the multi-round Shapley value.}
Figure~\ref{fig:illustrate} compares the proposed privacy-preserving CE scores with the $\MRSV$. Although the absolute values differ across schemes, the relative client rankings remain largely consistent. This demonstrates that our linear-time methods, despite considering far fewer coalitions than $\MRSV$ (that is incompatible with SA), can still effectively capture the ranking of clients. Table~\ref{tab:error} provides a more detailed comparison of the proposed CE scores, using both absolute error ($L_2$) and ranking correlation metrics (Spearman, Kendall, and Pearson). Both $\FP$ and $\EE$ yield rankings that closely align with those obtained with the baseline $\MRSV$, whereas $\LOO$ exhibits a noticeably weaker correlation. 

{\small
\begin{table}[!t]
    \caption{The various metrics of the privacy-preserving marginal difference based contribution scores, along with $\mathtt{COS}$ for a setting with 9 clients ($\mu=0.5$) at the 10th communication round. }
    \label{tab:error}
    \centering
    \begin{tabular}{cc|c||c|c|c}
         &  & $\LOO$     & $\EE$   &$\FP$&$\mathtt{COS}$\\
        \hline
        \multirow{4}{*}{\rotatebox{90}{ISIC}}
         & $L_2$ & $1.057 \pm .432$  & $0.339 \pm .159$  & $0.112\pm.168$ & $2.028\pm .807$\\
        & $\phi$   & $0.821\pm .042$ & $0.942\pm .054$ & $0.933\pm .045$ & $0.002\pm .667$\\
        & $\rho$   & $0.833\pm .051$    & $0.964 \pm .019$  &$0.961 \pm .034$ & $0.349\pm.398$\\
        & $\kappa$ & $0.783\pm .082$   & $0.901 \pm .024$   &$0.912 \pm .022$ & $0.098\pm.235$\\
        \hline 
        \multirow{4}{*}{\rotatebox{90}{Brain}} 
         & $L_2$ & $0.612 \pm .261$     & $0.163 \pm .074$   &$0.032 \pm .021$ & $0.964\pm .363$\\
        & $\phi$   & $0.793\pm .111$  & $0.891 \pm .063$ &$0.891 \pm .063$ & $0.083\pm .065$ \\
        & $\rho$   & $0.831 \pm .022$    & $0.911 \pm .012$  &$0.916 \pm .013$ & $0.211\pm.032$\\
        & $\kappa$ & $0.782 \pm .112$     & $0.871 \pm .113$  &$0.871 \pm .113$ & $0.001\pm.052$\\
        \hline 
        \multirow{4}{*}{\rotatebox{90}{PCam}} 
         & $L_2$ & $0.667 \pm .096$ & $0.422 \pm .093$   & $0.396\pm .053$ & $1.690\pm .206$\\
        & $\phi$   & $0.546\pm .149$     & $0.946 \pm .035$  &$0.944 \pm .024$ & $0.211\pm.376$\\
        & $\rho$   & $0.645\pm .002$  & $0.965 \pm .004$ &$0.968 \pm .002$ & $0.308\pm.024$ \\
        & $\kappa$ & $0.476\pm .133$    & $0.901 \pm .086$  &$0.901 \pm .055$ & $0.148\pm .087$\\
        \hline 
        \multirow{4}{*}{\rotatebox{90}{CIFAR10}} 
         & $L_2$ & $1.607 \pm .533$    & $0.077 \pm .045$   &$0.033 \pm .014$ & $1.234\pm .106$\\
        & $\phi$   & $0.684\pm .203$ & $0.904\pm .042$ &  $0.904\pm .055$ & $0.004\pm .322$ \\
        & $\rho$   & $0.821 \pm .114$    & $0.951 \pm .024$  &$0.955 \pm .022$ & $0.012\pm.303$\\
        & $\kappa$ & $0.691 \pm .082$    & $0.875 \pm .043$   &$0.869 \pm .046$ & $0.013\pm.273$\\ 
    \end{tabular}
\end{table}}

{\small
\begin{table}[!b]
\caption{The $L_2$ error and coefficient $\phi$ of our privacy-preserving marginal difference based scores compared to the true Shapley value for a setting with $9$ clients ($\mu = 0.5$) at the 10th communication round. }
    \label{tab:error_game}
    \centering
    \begin{tabular}{cc|c|c|c|c!}
         &  & $\MRSV$  &$\LOO$ & $\EE$ & $\FP$\\
         \hline
        \multirow{2}{*}{\rotatebox{90}{\scriptsize{ISIC}}} 
         & $L_2$ & $0.024\pm .029$ & $1.384\pm .051$ & $0.649\pm .011$  & $0.459\pm .018$\\
        & $\phi$ & $0.990\pm .002$ & $0.741\pm .004$  & $0.867\pm .062$& $0.861\pm .073$ \\
        \hline 
        \multirow{2}{*}{\rotatebox{90}{\scriptsize{Brain}}} 
         & $L_2$ & $0.191\pm .053$ & $1.660\pm 0.410$ & $0.775\pm .070$  &$0.205\pm .023$\\
        & $\phi$ & $0.982\pm .012$ & $0.517\pm .142$  & $0.971\pm .023$& $0.975\pm .023$\\
        \hline 
        \multirow{2}{*}{\rotatebox{90}{\scriptsize{PCam}}} 
        & $L_2$ & $0.161\pm .028$ & $1.843\pm .079$ & $0.308\pm .049$  &$0.304\pm .061$\\
        & $\phi$ & $0.966\pm .038$ & $0.749\pm .051$  & $0.837\pm .013$& $0.837\pm .013$\\
        \hline 
        \multirow{2}{*}{\rotatebox{90}{\tiny{CIFAR10}}} 
         & $L_2$ & $0.025\pm .006$ & $1.850\pm .387$ & $0.352\pm .053$ &$0.137\pm .029$\\
        & $\phi$ & $0.996\pm .004$ & $0.483\pm .319$  & $0.841\pm .035$ & $0.831\pm .045$\\
         \hline 
    \end{tabular}
\end{table}}

\paragraph{Comparison to cosine similarity.}
For completeness, we also include a comparison with an alternative privacy-preserving method based on cosine similarity, a widely used approach that does not rely on marginal differences~\cite{evans2024data,wu2021fast,xu2021gradient}. In this method, each client is assigned a score based on the similarity between its local gradient and the aggregated global gradient, i.e., $\COS(i) = \text{CosSim}(M_0 + U_i, M)$. This method also operates only with the coalitions available with SA. The last column of Table~\ref{tab:error}  presents the performance of the $\COS$ score with respect to $\MRSV$ across all evaluation metrics. Both our proposed scoring schemes consistently and significantly outperform $\COS$.

%\paragraph{Comparison to the true Shapley value.} In Table~\ref{tab:error_game}, we also compare against the true Shapley value, which is based on the retraining game, where each coalition of clients trains a model from scratch using FL, rather than aggregating gradients at each round. The proposed CE scores exhibit a high correlation ($\phi$ close to 1) with the true Shapley value.

\paragraph{Comparison to the true Shapley value.} In Table~\ref{tab:error_game}, we also compare against the true Shapley value, which is based on the retraining game, where each coalition of clients trains a model from scratch using FL, rather than aggregating gradients at each round. The proposed CE scores exhibit a high correlation ($\phi$ close to 1) with the true Shapley value, demonstrating that they effectively capture client contributions across diverse settings, including non-IID distributions and varying client counts. This indicates that our approximations preserve the relative importance of clients even when exact Shapley computation would be infeasible, providing a practical yet principled alternative for real-world deployments.

\subsection{Ablation Study}

%We next analyze the impact of three factors on CE: the training round at which the evaluation is performed, the number of clients, and the distribution of data across clients. Table~\ref{tab:round} reports the correlation of our proposed scores with $\MRSV$  when the evaluation is carried out five rounds earlier or later than in the default setting.  Table~\ref{tab:client} shows the correlation for varying numbers of clients. Finally, Table~\ref{tab:alpha} presents the performance of scores relative to $\LOO$ for different values of $\mu$ and the non-IID Dirichlet partitioning of the data. 

We next analyze the impact of three factors on CE: the training round at which the evaluation is performed, the number of clients, and the distribution of data across clients. Table~\ref{tab:round} reports the correlation of our proposed scores with $\MRSV$ when the evaluation is carried out five rounds earlier or later than in the default setting.  Table~\ref{tab:client} shows the correlation for varying numbers of clients. Finally, Table~\ref{tab:alpha} presents the performance of scores relative to $\LOO$ for different values of $\mu$ and the non-IID Dirichlet partitioning of the data. These results indicate that our scoring schemes remain reliable under realistic variations in training dynamics, client populations, and data heterogeneity, confirming their robustness for practical cross-silo deployments.

{\small
\begin{table}[!t]
      \caption{The Spearman coefficient of $\MRSV$ with $\LOO$, $\EE$,  and $\FP$ with 9 clients ($\mu=0.5$) at different communication rounds: 5, 10, 15.}
    \label{tab:round}
    \begin{tabular}{cc|c|c|c!}
         & & $\LOO$  & $\EE$  &$\FP$\\
         \hline  
         \multirow{3}{*}{\rotatebox{90}{ISIC}} 
          & 5& $0.735\pm .022$ & $0.843\pm .080$  &$0.842\pm .074$\\
         & 10   & $0.821\pm .042$ & $0.942\pm .054$ & $0.933\pm .045$\\
         & 15 & $0.905\pm .006$ & $0.986\pm .001$ &  $0.986\pm .001$\\
          \hline 
         \multirow{3}{*}{\rotatebox{90}{Brain}} 
          & 5 & $0.731 \pm .199$ & $0.881 \pm .088$ &$0.880\pm .093$ \\
         & 10 & $0.793\pm .111$  & $0.891 \pm .063$ &$0.891 \pm .063$\\
         & 15 & $0.900\pm .102$ & $0.986\pm .002$ &  $0.983\pm .009$\\
         \hline 
         \multirow{3}{*}{\rotatebox{90}{PCam}} 
          & 5 & $0.311\pm .285$ & $0.866 \pm .178$ &$0.869\pm .160$ \\
         & 10  & $0.546\pm .149$     & $0.946 \pm .035$  &$0.944 \pm .024$\\
         & 15 & $0.803\pm .011$ & $0.952\pm .004$ &  $0.952\pm .004$\\
         \hline  
         \multirow{3.3}{*}{\rotatebox{90}{CIFAR10}} 
          & 5  & $0.606 \pm .034$  & $0.784 \pm .147$   &$0.785 \pm .137$ \\
         & 10 & $0.684\pm .203$ & $0.904\pm .042$ &  $0.904\pm .055$\\
         & 15 & $0.829\pm .164$ & $0.973\pm .023$ &  $0.973\pm .023$\\
         \hline  
    \end{tabular}
    
    \caption{The Spearman coefficient of $\MRSV$ with $\LOO$,  $\EE$,  and $\FP$ when the data is split into 6, 9, 12, and 15 clients ($\mu=0.5$) at the 10th communication round.}
    \label{tab:client}
    \begin{tabular}{cc|c|c|c!}
         & & $\LOO$  & $\EE$  &$\FP$\\
         \hline  
         \multirow{4}{*}{\rotatebox{90}{ISIC}} 
          & 6& $0.862\pm .001$ & $0.941\pm .015$  &$0.940\pm .013$\\
         & 9 & $0.821\pm .042$ & $0.942\pm .054$ & $0.933\pm .045$\\
         & 12 & $0.733\pm .022$ & $0.882\pm .020$ &  $0.890\pm .011$\\
         & 15 & $0.682\pm .031$ & $0.864\pm .078$ & $0.870\pm .051 $\\
         \hline 
         \multirow{4}{*}{\rotatebox{90}{Brain}} 
           & 6 & $0.932 \pm .012$ & $0.950 \pm .050$ &$0.961\pm .024$ \\
         & 9 & $0.793\pm .111$  & $0.891 \pm .063$ &$0.891 \pm .063$\\
         & 12 & $0.745\pm .083$ & $0.844\pm .103$ &  $0.854\pm .141$\\
         & 15 & $0.612\pm .243$ & $0.841\pm .028$ & $0.841\pm .051 $\\
         \hline 
         \multirow{4}{*}{\rotatebox{90}{PCam}} 
          & 6 & $0.737\pm .005$ & $0.980 \pm .006$ &$0.981\pm .005$ \\
         & 9 & $0.546\pm .149$     & $0.946 \pm .035$  &$0.944 \pm .024$\\
         & 12 & $0.528\pm .032$ & $0.905\pm .077$ &  $0.906\pm .035$\\
         & 15 & $0.512\pm .006$ & $0.866\pm .033$ & $0.866\pm .054 $\\
         \hline 
         \multirow{4}{*}{\rotatebox{90}{CIFAR10}} 
          & 6  & $0.793 \pm .151$  & $0.955 \pm .023$   &$0.948 \pm .021$ \\
         & 9 & $0.684\pm .203$ & $0.904\pm .042$ &  $0.904\pm .055$\\
         & 12 & $0.601\pm .094$ & $0.892 \pm .072$ &$0.892 \pm .072$\\
         & 15 & $0.431\pm .006$ & $0.862\pm .005$ & $0.868\pm .004 $\\
         \hline 
    \end{tabular}
    
    \caption{The Spearman coefficient of $\MRSV$ with  $\EE$, and $\FP$ with 9 clients ($\mu\in\{0.1, 0.5, 1.0)\}$) at the 10th communication round.}
    \label{tab:alpha}
    \begin{tabular}{cc|c|c|c!}
         & & $\LOO$  & $\EE$  &$\FP$\\
         \hline  
         \multirow{3}{*}{\rotatebox{90}{ISIC}} 
         & 0.1& $0.868\pm .013$ & $0.981\pm .004$  &$0.980\pm .006$\\
         & 0.5  & $0.821\pm .042$ & $0.942\pm .054$ & $0.933\pm .045$\\
         & 1.0 & $0.822\pm .022$ & $0.916\pm .032$ &  $0.916\pm .056$\\
         \hline 
         \multirow{3}{*}{\rotatebox{90}{Brain}} 
          & 0.1 & $0.801\pm .032$ & $0.945 \pm .037$ &$0.935 \pm .039$ \\
         & 0.5 & $0.793\pm .111$  & $0.891 \pm .063$ &$0.891 \pm .063$\\
         & 1.0 & $0.677\pm .147$ & $0.817\pm .139$ &  $0.816\pm .123$\\
         \hline  
         \multirow{3}{*}{\rotatebox{90}{PCam}} 
          & 0.1 & $0.757\pm 013$ & $0.980 \pm .023$ &$0.981 \pm .012$ \\
         & 0.5 & $0.546\pm .149$     & $0.946 \pm .035$  &$0.944 \pm .024$\\
         & 1.0 & $0.534\pm .076$ & $0.924\pm .038$ &  $0.931\pm .021$\\
         \hline 
         \multirow{3.3}{*}{\rotatebox{90}{CIFAR10}} 
          & 0.1  & $0.880 \pm .014$  & $0.963 \pm .033$   &$0.965 \pm .033$ \\
         & 0.5 & $0.684\pm .203$ & $0.904\pm .042$ &  $0.904\pm .055$\\
         & 1.0 & $0.648\pm .021$ & $0.850\pm .042$ &  $0.848\pm .046$\\
         \hline   
    \end{tabular}
\end{table}}

 As expected, performance improves in later rounds for nearly all methods, since contributions are easier to assess once the model has converged. In contrast, increasing the number of clients leads to reduced correlation, reflecting the widening gap between the exponentially complex Shapley value and our linear-time approximations. Nevertheless, both $\FP$ and $\EE$ stay stable, maintaining strong performance even under more challenging conditions such as early-round evaluation or larger client populations. 
 %Similar trends are observed in the other two tables. CE is more accurate with fewer clients, while it degrades as the number of clients increases.
Regarding data distribution, higher values of $\mu$ lead to reduced performance of the scoring schemes. This behavior is expected, since larger $\mu$ produces distributions closer to IID, making clients harder to distinguish. Nevertheless, $\FP$ and $\EE$ remain robust across all three aspects studied.

\subsection{Robustness}

Next, we analyze the impact of selfish behavior in the post-processing phase, focusing on scenarios where a client misreports its marginal contribution to gain an advantage. While $\FP$ and $\LOO$ are vulnerable to this type of manipulation, recall that clients cannot influence their own scores in the $\EE$ scheme. On the other hand, according to \eqref{eq:EE}, client $i$ can influence the score of others to the extent of $\big(v(M) - v(M_0 + U_i)\big) + \big(v(M - U_i) - v(M_0)\big)$. We present the relative, percentage-wise influence in Table~\ref{tab:attack}.

%i.e.,  this interaction by reporting the ratios below. 

%$$
%\frac{\big(v(M) - v(M_0 + U_i)\big) + \big(v(M - U_i) - v(M_0)\big)}{\EE(j)}
%$$

The main diagonal of the table is zero, since no client contributes to its own score. This demonstrates that $\EE$ is resistant to selfish behavior. Note that it remains vulnerable to strategic misreporting: a client could artificially deflate its reported value to diminish other clients' scores. This vulnerability can be mitigated by augmenting the averaging step in $\EE$'s definition~\eqref{eq:LEEO-IEEI} with anomaly detection techniques or by adopting Byzantine-resilient aggregation rules (e.g., substituting the summation with a median operator). A similar attack scenario under analogous trust assumptions was explored in the context of decentralized FL in~\cite{anada2025measuring}.

\begin{table}[!t]
\caption{The relative influence (e.g., the columns sum up to 1) of a single client on another's $\EE$ score using the Brain dataset with 6 clients ($\mu=0.5$) at the 10th communication round. We only show the mean as the variance is always below $0.01$.}
    \label{tab:attack}
    \centering
    \begin{tabular}{c|c|c|c|c|c|c}
        No. clients & 1 & 2 & 3 & 4 & 5 & 6 \\
       \hline  
        1 & $0.000$ & $0.145$ & $0.151$ & $0.169$ & $0.150$ & $0.167$ \\
        2 & $0.129$ & $0.000$ & $0.132$ & $0.147$ & $0.131$ & $0.145$ \\
        3 & $0.170$ & $0.167$ & $0.000$ & $0.194$ & $0.172$ & $0.192$ \\
        4 & $0.273$ & $0.268$ & $0.280$ & $0.000$ & $0.278$ & $0.309$ \\
        5 & $0.164$ & $0.161$ & $0.167$ & $0.187$ & $0.000$ & $0.185$ \\
        6 & $0.262$ & $0.257$ & $0.268$ & $0.300$ & $0.266$ & $0.000$ \\
        \hline 
    \end{tabular}
\end{table}

\subsection{Downstream tasks}

\paragraph{Performance Gain. }
To demonstrate practical utility beyond ranking alignment, we evaluate global model performance when aggregating client updates weighted by their CE scores. Specifically, we assess the effectiveness of our proposed scoring schemes when clients contribute data of varying quality, as in~\cite {pejo2023quality}. We consider IID clients with different levels of label noise: labels in each client's dataset are randomized with a linearly increasing probability before training. With 9 clients, the label-flipping rates are [0.00, 0.11, 0.22, \dots, 0.89, 1.00]
. During aggregation, each update is weighted by the current value of its normalized score, so that top-performing and bottom-performing participants receive weights above and below 1, respectively.

In Figure~\ref{fig:perf_cifar_isic}, it is visible that both $\FP$ and $\EE$ are superior to FedAvg (naive uniform aggregation) and to the baselines $\COS$ and $\LOO$. In fact, their performance almost reaches that of $\MRSV$. These results imply that 1) even the limited information available with SA could be used beneficially, and 2) the potential highest gain of client weighting is almost fully reachable, even with privacy in mind. Hence, both $\FP$ and $\EE$ successfully identify and upweight high-quality contributors while suppressing noise, translating their valuation information into tangible gains in model performance.

\input{tikZ}

\paragraph{Misbehavior Detection. }
Another frequently utilized benchmarking downstream task is the identification of dishonest participants. Again, we consider honest IID clients with a single exception: one is a Byzantine client executing an adversarial attack (label flipping). 

In Figure~\ref{fig:misbeh_cifar_isic}, we measure the detection rate of the scoring schemes by counting how often the attacker received the lowest score (out of ten experiments). It is visible that both $\FP$ and $\EE$ are superior to $\LOO$ and $\COS$, almost reaching the upper bound of $\MRSV$. These results imply that privacy-preserving misbehavior detection is practical. 

\section{Conclusion}
\label{sec:con}

We study contribution evaluation in cross-silo federated learning under secure aggregation, where individual client updates must remain hidden while reward allocation and accountability still require meaningful per-client scoring. This setting creates a fundamental mismatch: classical marginal contribution methods (including Shapley-style approaches) require evaluating many coalitions and accessing individual updates, whereas secure aggregation only exposes the empty and grand coalitions and prevents the server from inspecting client updates. To bridge this gap, we propose two privacy-preserving marginal difference scores fully compatible with secure aggregation: Fair-Private (FP), which leverages the limited coalitions available to each client to satisfy standard fairness axioms, and Everybody-Else (EE), which avoids self evaluation and provides manipulation resistance through indirect cross-scoring. We establish formal guarantees for fairness, privacy, and computational efficiency, and empirically demonstrate across multiple medical image datasets and CIFAR10 that FP and EE substantially outperform existing secure-aggregation-compatible baselines, closely approximating Shapley-induced rankings while remaining practical at scale.

Overall, our results indicate that FP and EE not only preserve Shapley-based fairness principles but also produce stable, interpretable client contributions under challenging federated learning conditions, including non-IID data and large client populations. This demonstrates that meaningful client valuations can be obtained efficiently and privately, while also supporting reward allocation and misbehavior detection. Their compatibility with standard federated frameworks enables deployment without costly architectural changes or added overhead.

Looking forward, key directions include extending private contribution evaluation to settings with client sampling, churn, or asynchronous training, while maintaining secure-aggregation compatibility, and further strengthening resistance to strategic manipulation by combining EE-style cross-evaluation with robust aggregation or anomaly detection to mitigate adversarial influence.
%We addressed the fundamental challenge of fairly and privately evaluating client contributions in federated learning under secure aggregation. We proposed two novel contribution evaluation methods that overcome the limitations of existing approaches: a privacy-preserving contribution evaluation score that relies on self evaluation but satisfies all desired fairness properties, and a score that relies on evaluations from other clients, thereby enhancing robustness to manipulation in selfish settings.

%We established theoretical guarantees for the desired properties of both schemes. Empirical evaluation on medical image datasets demonstrates that our scores closely approximate the Shapley value rankings, while significantly outperforming existing privacy preserving alternatives in accuracy. These results show that fairness, privacy, and even robustness can be reconciled in contribution evaluation for federated learning without resorting to computationally expensive techniques.

\subsection*{Acknowledgments}
Project no. 145832, implemented with the support provided by the Ministry of Innovation and Technology from the NRDI Fund, financed under the PD\_23 funding scheme. 

\bibliographystyle{ACM-Reference-Format}
\bibliography{ref}

%\newpage
\appendix

\section{Proof of Theorem~\ref{th:L1O-I1I_ax}}\label{sec:appTh1}

     If client $i$ is a null player, meaning it does not contribute to any coalition, then $v(M_0+U_i)=v(M_0)$ and $v(M-U_i)=v(M)$, implying that the $\alpha(i)=0$, and therefore $\FP(i)=0$. This means $\FP$ satisfies property \ref{ax:null}.  Similarly, if two clients $i$, and $j$ contribute equally to all coalitions, then  $v(M-U_i)=v(M-U_j)$, and $v(M_0+U_i)=v(M_0+U_j)$, which implies the following:
     {\small
     \begin{align*}
         \alpha(i)=\cfrac{1}{2}((v(M) - v(M-U_i))-(v(M_0+U_i)-v(M_0)))\\
         =\cfrac{1}{2}((v(M) - v(M-U_j))-(v(M_0+U_j)-v(M_0)))\\
         =\alpha(j)
     \end{align*}}
     Thus, $\FP(i)=\FP(j)$, and $\FP$ satisfies the symmetry property \ref{ax:sym}. Finally, Property~\ref{ax:eff}, and~\ref{ax:priv} follows directly from the definitions of $\FP$, and $\alpha(i)$, since $\sum_{i=1}^{N}\FP(i)=v(M)$, and the scores of the clients depend only on the models corresponding to the permitted coalitions.

\begin{remark}\label{rem:edgecase}
While possible, the edge case in which the denominator in \eqref{eq:FP} is zero is disregarded in our definition. If such a situation arises, we propose to replace $\alpha(i)$ with either $v(M) - v(M-U_i)$ or $v(M_0+U_i)-v(M_0)$. This way, the score still satisfies the same properties, but it is highly unlikely that both sum to zero.
\end{remark}

\section{Proof of Theorem~\ref{th:LEEOIEEI}}
\label{sec:appTh2}

    We focus on the parts of the metric. Property~\ref{ax:priv} follows directly from the definitions of \(\beta(i)\) and \(\gamma(i)\), since these terms are composed of the available coalitions. For Property~\ref{ax:inc}, note that the only values a client $i$ could potentially misreport are $v(M - U_i)$ and $v(M_0 + U_i)$. However, these terms do not appear in \(\beta(i)\) or \(\gamma(i)\). Therefore, \(\EE\) satisfies Property~\ref{ax:inc}.  To prove symmetry, assume that clients $i$ and $j$ contribute equally to all possible coalitions, i.e.,  $v(M-U_i)=v(M-U_j)$ and $v(M_0+U_i)=v(M_0+U_j)$.  Then, based on \eqref{eq:LEEO-IEEI}:
    {\small
    \begin{align*}
        \beta(i)=\hspace{6.5cm}\\
        =\frac{
        \sum\limits_{k\not\in\{i,j\}}[v(M)-v(M_0+U_k)]+[v(M)-v(M_0+U_j)]}{(N-1)^2}\\
        =\frac{\sum\limits_{k\not\in\{i,j\}}[v(M)-v(M_0+U_k)]+[v(M)-v(M_0+U_i)]}{(N-1)^2}\\
        =\beta(j)
    \end{align*}
    \begin{align*}
        \gamma(i)=\hspace{6.5cm}\\
        =\frac{
        \sum\limits_{k\not\in\{i,j\}}[v(M-U_k)-v(M_0)]+[v(M-U_j)-v(M_0)]}{(N-1)^2}\\
        =\frac{\sum\limits_{k\not\in\{i,j\}}[v(M-U_k)-v(M_0)]+[v(M-U_i)-v(M_0)]}{(N-1)^2}\\
        =\gamma(j)
    \end{align*}}
    Thus, it is obtained $\EE(i)=\EE(j)$, this is $\EE$ satisfies symmetry property \ref{ax:sym}. Finally, Property~\ref{ax:eff} is also satisfied, as the same normalizing step is used as in $\FP$. Note that, for this reason, the same edge case could also emerge, as we detailed above. 
    
\section{Proof of Theorem~\ref{th:complex}}
\label{sec:appTh3}

We prove the statement for each individual building block of the metrics, ensuring that their combination also inherits this validity.
From Definition~\ref{def:FP2}, it is clear that computing $\alpha(i)$ for a single client requires $\mathcal{O}(1)$ operations. Thus, computing the scores for all $N$ clients requires $\mathcal{O}(N)$ time overall.
Similarly, according to \eqref{eq:LEEO-IEEI},  the complexity of computing the terms $\beta(i)$ and $\gamma(i)$ for a single client is $\mathcal{O}(N)$. However, the evaluations $v(M_0+U_j)$ and $v(M-U_j)$ for each client $j$ are shared when computing the scores for different clients $i$, since each evaluation depends only on client $j$, not on the client being scored. Therefore, the overall computational cost remains $\mathcal{O}(N)$.

%\section{$\EE$ score with Banzhaf-style normalization}
\label{sec:appEEb}

\section{Experimental Setup}
\label{app:Experiments}

{\footnotesize
\begin{table}[H]
    \centering
    \caption{Overview of the experimental setup}
    \label{tab:datasets}
    \begin{tabular}{l|c|c|c|c!}
        %\hline 
         & \textbf{ISIC2019} & \textbf{PCam} & \textbf{Brain-MRI} & \textbf{CIFAR10} \\
        \hline 
        \textbf{Task} & Multi-class & Binary-class & Multi-class & Multi-class \\
        \textbf{Pred.} & Melanoma & Metastatic & Tumor & Image \\
        \textbf{Model} & ResNet50 & ResNet18 &
        Simple CNN & Simple CNN \\
        \textbf{Size} & 23,247 & 327,680 & 7,023 & 60,000 \\
        \textbf{L-Epoch} & 3 & 5 & 5 & 5 \\
        \hline 
    \end{tabular}
\end{table}}

\section{The efficiency-linearity trade and semivalues}\label{app:semivalues}

The practical and conceptual shift introduced by federated learning brings new challenges. In particular, data valuation in federated learning aims not only to identify which sets of data points contribute more to the final model but also to satisfy clients’ fairness and privacy requirements. This shift is reflected in the increased importance of fairness properties in the federated learning setting. For instance, the efficiency property is significantly more relevant in federated learning than in the centralized scenario because contribution scores are assigned to clients rather than to individual data points. In this context, contribution evaluation is ultimately tied to resource allocation decisions (e.g., incentives, quotas, or priorities), which typically operate under a fixed budget. This stands in contrast to centralized data valuation, where the objective is simply to measure which data points improve model performance the most. Semivalues are appropriate when the goal is to describe players in a game without requiring a fair division~\cite{with_effic}. We view centralized data valuation as fitting this description.

As alternatives to the Shapley value, Beta Shapley and Data Banzhaf semivalues have been proposed to address data valuation problems in machine learning, particularly in settings such as noisy label detection, where the Shapley value has been shown to be suboptimal~\cite{kwon2021beta,data_banzhaf}. Specifically, a semivalue is a function in the space of games that satisfies the symmetry, null player, and linearity properties, but may not satisfy the efficiency property. 

Federated learning commonly requires dividing a fixed resource among participants, and in such scenarios, efficiency becomes a natural, and often necessary, requirement. We define our scores with this consideration in mind. For this reason, they lie outside the family of semivalues, with the exception of the unnormalized score $\FP$, i.e., $\alpha(i)$ (see Equation~\ref{eq:FP}). In contrast to the role of efficiency in federated learning, linearity can be relaxed in this setting. Linearity is primarily an algebraic and technical condition in cooperative game theory rather than a fairness axiom. Its main role is to establish uniqueness results or, in the case of semivalues, to characterize their functional form. Since our objective is not to prove the existence of unique solutions under a specific set of axioms, relaxing linearity does not compromise fairness guarantees or practical interpretability. As a result, analyzing the unnormalized score satisfying linearity offers limited additional insight from either a fairness or an applied perspective.

Overall, this discussion highlights that the choice of axioms underlying a data valuation score cannot be decoupled from the structure of the underlying machine learning task. While centralized settings may favor descriptive measures that relax efficiency, federated learning inherently demands allocation-aware scores that respect budget constraints. Understanding this interplay between axiomatic properties and application-specific requirements is, therefore, a central component of the data valuation endeavor.
\end{document}

%% file: tikZ.tex
% ================== PLOT 1==================
\begin{figure*}[!t]
%\centering
\vspace{-3mm}
\caption{The negative loss of the global model on the CIFAR10 (left), and ISIC2019 (right) datasets. Training is for $9$ clients with $10$ communication rounds, and the results of mean and variance are reported for the 10th run.}
\label{fig:perf_cifar_isic}
\hspace{0.75cm}
% --------------------  -------------------
\begin{minipage}[t]{0.42\textwidth}
\centering
\begin{tikzpicture}
\begin{axis}[
    width=\linewidth,
    height=6.2cm,
    xlabel={Communication round},
    ymajorgrids=false,
    xmajorgrids=false,
    grid style=dashed,
    legend=false
]

% =====================================================
% Curve blue
% =====================================================
\addplot[name path=mean5, thick, blue]
coordinates {
(0,-2.316019) (1,-2.020275) (2,-1.833844) (3,-1.702725) (4,-1.655008)
(5,-1.635281) (6,-1.630719) (7,-1.635935) (8,-1.664870) (9,-1.668974)
};

\addplot[name path=upper5, draw=none]
coordinates {
(0,-2.315394) (1,-1.987080) (2,-1.800701) (3,-1.698899) (4,-1.648382)
(5,-1.628225) (6,-1.627287) (7,-1.630520) (8,-1.659320) (9,-1.662161)
};

\addplot[name path=lower5, draw=none]
coordinates {
(0,-2.316643) (1,-2.053469) (2,-1.866986) (3,-1.706550) (4,-1.661635)
(5,-1.642337) (6,-1.634150) (7,-1.641350) (8,-1.670419) (9,-1.675787)
};

\addplot[blue!30, fill opacity=0.35]
fill between[of=upper5 and lower5];

% =====================================================
% Curve 2 (orange)
% =====================================================
\addplot[name path=mean2, thick, orange]
coordinates {
(0,-2.307204) (1,-2.296447) (2,-2.246389) (3,-2.184027) (4,-2.156139)
(5,-2.110454) (6,-2.079518) (7,-2.102336) (8,-2.073489) (9,-2.039742)
};

\addplot[name path=upper2, draw=none]
coordinates {
(0,-2.307197) (1,-2.296158) (2,-2.244197) (3,-2.181412) (4,-2.154804)
(5,-2.109113) (6,-2.078347) (7,-2.099904) (8,-2.070514) (9,-2.038675)
};

\addplot[name path=lower2, draw=none]
coordinates {
(0,-2.307212) (1,-2.296737) (2,-2.248582) (3,-2.186641) (4,-2.157475)
(5,-2.111794) (6,-2.080689) (7,-2.104767) (8,-2.076464) (9,-2.040809)
};

\addplot[orange!30, fill opacity=0.35]
fill between[of=upper2 and lower2];

% =====================================================
% Curve 3 (green)
% =====================================================
\addplot[name path=mean3, thick, green!60!black]
coordinates {
(0,-2.302813) (1,-2.137063) (2,-2.004226) (3,-1.905013) (4,-1.900394)
(5,-1.889672) (6,-1.844544) (7,-1.878087) (8,-1.890657) (9,-1.921745)
};

\addplot[name path=upper3, draw=none]
coordinates {
(0,-2.302796) (1,-2.123550) (2,-1.975930) (3,-1.857367) (4,-1.847892)
(5,-1.831341) (6,-1.783771) (7,-1.823073) (8,-1.840111) (9,-1.878592)
};

\addplot[name path=lower3, draw=none]
coordinates {
(0,-2.302831) (1,-2.150575) (2,-2.032523) (3,-1.952660) (4,-1.952896)
(5,-1.948004) (6,-1.905316) (7,-1.933101) (8,-1.941202) (9,-1.964898)
};

\addplot[green!30, fill opacity=0.35]
fill between[of=upper3 and lower3];

% =====================================================
% Curve 4 (red)
% =====================================================
\addplot[name path=mean4, thick, red]
coordinates {
(0,-2.301046) (1,-1.999800) (2,-1.759503) (3,-1.671811) (4,-1.640208)
(5,-1.631818) (6,-1.628247) (7,-1.659148) (8,-1.666336) (9,-1.676239)
};

\addplot[name path=upper4, draw=none]
coordinates {
(0,-2.299836) (1,-1.959284) (2,-1.736402) (3,-1.666916) (4,-1.633050)
(5,-1.627136) (6,-1.624391) (7,-1.652657) (8,-1.661709) (9,-1.670109)
};

\addplot[name path=lower4, draw=none]
coordinates {
(0,-2.302256) (1,-2.040316) (2,-1.782605) (3,-1.676707) (4,-1.647367)
(5,-1.636499) (6,-1.632104) (7,-1.665639) (8,-1.670962) (9,-1.682369)
};

\addplot[red!30, fill opacity=0.35]
fill between[of=upper4 and lower4];

% =====================================================
% Curve  (purple)
% =====================================================
\addplot[name path=mean1, thick, purple]
coordinates {
(0,-2.308436) (1,-1.990556) (2,-1.760540) (3,-1.689985) (4,-1.651580)
(5,-1.667711) (6,-1.667302) (7,-1.700643) (8,-1.713370) (9,-1.7231140)
};

\addplot[name path=upper1, draw=none]
coordinates {
(0,-2.308260) (1,-1.950240) (2,-1.741003) (3,-1.660255) (4,-1.632210)
(5,-1.655264) (6,-1.651864) (7,-1.684728) (8,-1.708721) (9,-1.718312)
};

\addplot[name path=lower1, draw=none]
coordinates {
(0,-2.308612) (1,-2.030873) (2,-1.780076) (3,-1.719716) (4,-1.670950)
(5,-1.680158) (6,-1.682739) (7,-1.716559) (8,-1.718020) (9,-1.729967)
};

\addplot[purple!30, fill opacity=0.35]
fill between[of=upper1 and lower1];

% =====================================================
% Curve 6 (violet)
% =====================================================
\addplot[name path=mean6, thick, violet]
coordinates {
(0,-2.308197) (1,-2.305030) (2,-2.256275) (3,-2.252607) (4,-2.250978)
(5,-2.193883) (6,-2.211439) (7,-2.209964) (8,-2.175933) (9,-2.178889)
};

\addplot[name path=upper6, draw=none]
coordinates {
(0,-2.308186) (1,-2.305018) (2,-2.252926) (3,-2.249914) (4,-2.248691)
(5,-2.185368) (6,-2.201149) (7,-2.200777) (8,-2.161261) (9,-2.168674)
};

\addplot[name path=lower6, draw=none]
coordinates {
(0,-2.308207) (1,-2.305041) (2,-2.259624) (3,-2.255300) (4,-2.253265)
(5,-2.202398) (6,-2.221729) (7,-2.219152) (8,-2.190604) (9,-2.189105)
};

\addplot[violet!30, fill opacity=0.35]
fill between[of=upper6 and lower6];

\end{axis}
\end{tikzpicture}
\end{minipage}
%\hspace{0.4cm}
% ================== PLOT  1==================
\begin{minipage}[t]{0.42\textwidth}
\centering
\begin{tikzpicture}
\begin{axis}[
    width=\linewidth,
    height=6.2cm,
    xlabel={Communication round},
    ymajorgrids=true,
    xmajorgrids=false,
    grid style=dashed,
    legend=false
]

% Curve 1 (blue)
% =====================================================
\addplot[name path=mean4, thick, blue]
coordinates {
(0,-2.321584) (1,-2.144289) (2,-1.892157) (3,-1.835221) (4,-1.808081)
(5,-1.778006) (6,-1.781441) (7,-1.792060) (8,-1.803350) (9,-1.821012)
};

\addplot[name path=upper4, draw=none]
coordinates {
(0,-2.321519) (1,-2.113121) (2,-1.858206) (3,-1.783741) (4,-1.748140)
(5,-1.716400) (6,-1.717750) (7,-1.722213) (8,-1.734690) (9,-1.769339)
};

\addplot[name path=lower4, draw=none]
coordinates {
(0,-2.321649) (1,-2.175456) (2,-1.926108) (3,-1.886701) (4,-1.868022)
(5,-1.839613) (6,-1.845133) (7,-1.861908) (8,-1.872011) (9,-1.872685)
};

\addplot[blue!30, fill opacity=0.35]
fill between[of=upper4 and lower4];

%=====================================================
% Curve 2 (orange)
% =====================================================
\addplot[name path=mean2, thick, orange]
coordinates {
(0,-2.306302) (1,-2.297072) (2,-2.275994) (3,-2.266683) (4,-2.276340)
(5,-2.251951) (6,-2.258646) (7,-2.266636) (8,-2.245621) (9,-2.237568)
};

\addplot[name path=upper2, draw=none]
coordinates {
(0,-2.306294) (1,-2.296892) (2,-2.273138) (3,-2.261478) (4,-2.273559)
(5,-2.241623) (6,-2.250848) (7,-2.261422) (8,-2.232540) (9,-2.220579)
};

\addplot[name path=lower2, draw=none]
coordinates {
(0,-2.306311) (1,-2.297252) (2,-2.278850) (3,-2.271888) (4,-2.279120)
(5,-2.262279) (6,-2.266444) (7,-2.271850) (8,-2.258702) (9,-2.254557)
};

\addplot[orange!30, fill opacity=0.35]
fill between[of=upper2 and lower2];

% =====================================================
% =====================================================
% =====================================================
% Curve 3 (green)
% =====================================================
\addplot[name path=mean3, thick, green!60!black]
coordinates {
(0,-2.305704) (1,-2.215846) (2,-2.187688) (3,-2.149644) (4,-2.131883)
(5,-2.114022) (6,-2.116360) (7,-2.103584) (8,-2.100804) (9,-2.112131)
};

\addplot[name path=upper3, draw=none]
coordinates {
(0,-2.225702) (1,-2.215356) (2,-2.185203) (3,-2.145894) (4,-2.126135)
(5,-2.105363) (6,-2.106678) (7,-2.093210) (8,-2.088462) (9,-2.103219)
};

\addplot[name path=lower3, draw=none]
coordinates {
(0,-2.225706) (1,-2.216336) (2,-2.190172) (3,-2.153394) (4,-2.137631)
(5,-2.122682) (6,-2.126043) (7,-2.113959) (8,-2.113147) (9,-2.121043)
};

\addplot[green!30, fill opacity=0.35]
fill between[of=upper3 and lower3];

% =====================================================
% Curve 4 (red)
% =====================================================
\addplot[name path=mean1, thick, red]
coordinates {
(0,-2.311920) (1,-2.182142) (2,-1.964947) (3,-1.870282) (4,-1.891720)
(5,-1.867128) (6,-1.841217) (7,-1.845066) (8,-1.881905) (9,-1.885857)
};

\addplot[name path=upper1, draw=none]
coordinates {
(0,-2.311745) (1,-2.155065) (2,-1.894768) (3,-1.774893) (4,-1.805401)
(5,-1.778967) (6,-1.745562) (7,-1.764206) (8,-1.806764) (9,-1.819072)
};

\addplot[name path=lower1, draw=none]
coordinates {
(0,-2.312094) (1,-2.209220) (2,-2.035125) (3,-1.965671) (4,-1.978038)
(5,-1.955288) (6,-1.936873) (7,-1.925927) (8,-1.957046) (9,-1.952643)
};

\addplot[red!30, fill opacity=0.35]
fill between[of=upper1 and lower1];

% 

% =====================================================
% Curve 5 (purple)
% =====================================================
\addplot[name path=mean5, thick, purple]
coordinates {
(0,-2.310250) (1,-2.014840) (2,-1.850614) (3,-1.802264) (4,-1.769643)
(5,-1.789403) (6,-1.799213) (7,-1.824341) (8,-1.880498) (9,-1.880202)
};

\addplot[name path=upper5, draw=none]
coordinates {
(0,-2.309714) (1,-1.970263) (2,-1.773450) (3,-1.719462) (4,-1.686468)
(5,-1.709167) (6,-1.716055) (7,-1.745923) (8,-1.796141) (9,-1.822558)
};

\addplot[name path=lower5, draw=none]
coordinates {
(0,-2.310787) (1,-2.059417) (2,-1.927778) (3,-1.885066) (4,-1.852817)
(5,-1.869639) (6,-1.882372) (7,-1.902759) (8,-1.964856) (9,-1.992245)
};

\addplot[purple!30, fill opacity=0.35]
fill between[of=upper5 and lower5];

% =====================================================
% Curve 6 (violet)
% =====================================================
\addplot[name path=mean6, thick, violet]
coordinates {
(0,-2.305704) (1,-2.295846) (2,-2.267688) (3,-2.229644) (4,-2.211883)
(5,-2.194022) (6,-2.196360) (7,-2.183584) (8,-2.180804) (9,-2.192131)
};

\addplot[name path=upper6, draw=none]
coordinates {
(0,-2.305702) (1,-2.295356) (2,-2.265203) (3,-2.225894) (4,-2.206135)
(5,-2.185363) (6,-2.186678) (7,-2.173210) (8,-2.168462) (9,-2.183219)
};

\addplot[name path=lower6, draw=none]
coordinates {
(0,-2.305706) (1,-2.296336) (2,-2.270172) (3,-2.233394) (4,-2.217631)
(5,-2.202682) (6,-2.206043) (7,-2.193959) (8,-2.193147) (9,-2.201043)
};

\addplot[violet!30, fill opacity=0.35]
fill between[of=upper6 and lower6];

\end{axis}
\end{tikzpicture}
\end{minipage}
\hfill
\vspace{1mm}
% ===================== LEGEND (BOTTOM) =====================
\begin{minipage}{0.6\textwidth}
\centering
\begin{tikzpicture}
    \draw[blue, thick] (0,0) -- (0.5,0);
    \node[right] at (0.6,0) {$\MRSV$};

    \draw[orange, thick] (2.0,0) -- (2.5,0);
    \node[right] at (2.6,0) {$\COS$};

    \draw[green!60!black, thick] (3.5,0) -- (4.0,0);
    \node[right] at (4.1,0) {$\LOO$};

    \draw[red, thick] (5.0,0) -- (5.5,0);
    \node[right] at (5.6,0) {$\FP$};

    \draw[purple, thick] (6.4,0) -- (6.9,0);
    \node[right] at (7.0,0) {$\EE$};

    \draw[violet, thick] (7.6,0) -- (8.1,0);
    \node[right] at (8.2,0) {$\mathtt{FedAvg}$};
\end{tikzpicture}
\end{minipage}
\end{figure*}
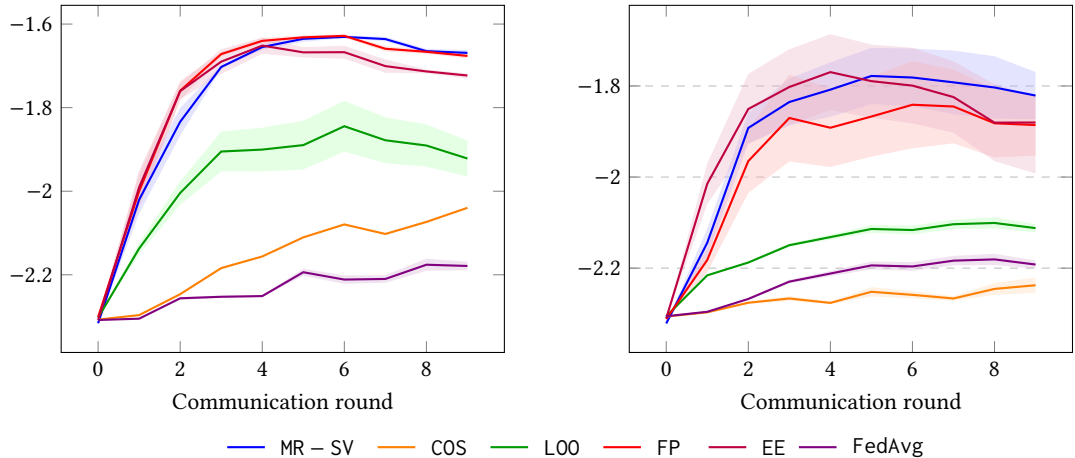

%\hspace{-4mm}
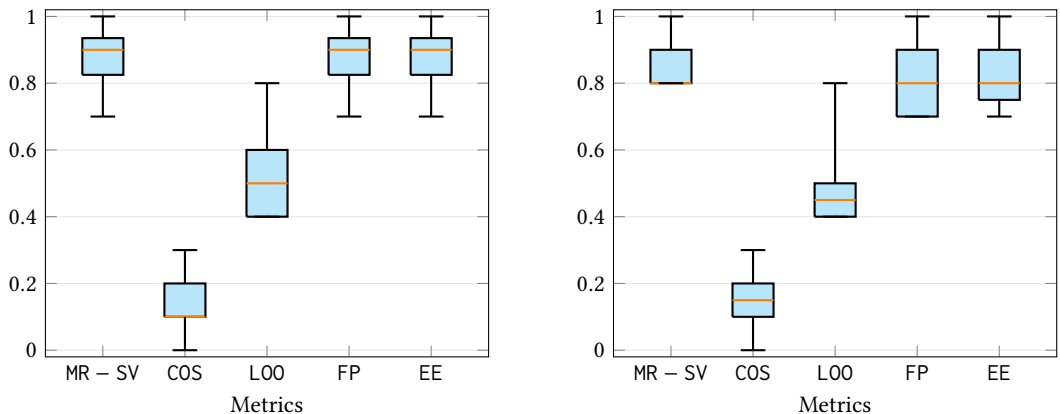
\begin{figure*}[!t]
%\centering
%\vspace{-1mm}
\caption{Misbehavior detection score (what is the probability that the attacker gets the lowest score) on the CIFAR10 (left), and ISIC2019 (right) datasets. Training is for $9$ clients with $10$ communication rounds, and results are reported for the 10th run. }
\label{fig:misbeh_cifar_isic}

\begin{minipage}[t]{0.42\textwidth}
%\centering
\begin{tikzpicture}
\begin{axis}[
    width=\linewidth,
    height=6.2cm,
    xlabel={Metrics},
    ymin=-0.02, ymax=1.02,
    ymajorgrids=true,
    xmajorgrids=false,
    grid style={draw=gray!20},
    xtick={1,2,3,4,5},
    xticklabels={$\MRSV$,$\COS$,$\LOO$,$\FP$,$\EE$},
]

\manualbox{1}{0.70}{0.825}{0.90}{0.935}{1.00}
\manualbox{2}{0.00}{0.10}{0.10}{0.20}{0.30}
\manualbox{3}{0.40}{0.40}{0.50}{0.60}{0.80}
\manualbox{4}{0.70}{0.825}{0.90}{0.935}{1.00}
\manualbox{5}{0.70}{0.825}{0.90}{0.935}{1.00}
\end{axis}
\end{tikzpicture}
\end{minipage}
%\hspace{0.2cm}
\begin{minipage}[t]{0.42\textwidth}
%\centering
\begin{tikzpicture}
\begin{axis}[
    width=\linewidth,
    height=6.2cm,
    xlabel={Metrics},
    ymin=-0.02, ymax=1.02,
    ymajorgrids=true,
    xmajorgrids=false,
    grid style={draw=gray!20},
    xtick={1,2,3,4,5},
    xticklabels={$\MRSV$,$\COS$,$\LOO$,$\FP$,$\EE$},
]
\manualbox{1}{0.80}{0.80}{0.80}{0.90}{1.00}
\manualbox{2}{0.00}{0.10}{0.15}{0.20}{0.30}
\manualbox{3}{0.40}{0.40}{0.45}{0.50}{0.80}
\manualbox{4}{0.70}{0.70}{0.80}{0.90}{1.00}
\manualbox{5}{0.70}{0.75}{0.80}{0.90}{1.00}
\end{axis}
\end{tikzpicture}
\end{minipage}
\end{figure*}

%% file: ref.bib
@article{with_effic,
 author = {Pradeep Dubey and Abraham Neyman and Robert James Weber},
 journal = {Mathematics of Operations Research},
 number = {1},
 pages = {122--128},
 title = {Value Theory without Efficiency},
 volume = {6},
 year = {1981}
}

@Inproceedings{Cosine_pitfalls,
title={The Hidden Pitfalls of the Cosine Similarity Loss},
author={Andrew Draganov and Sharvaree Vadgama and Erik J Bekkers},
booktitle={ICML workshop: High-dimensional Learning Dynamics 2024: The Emergence of Structure and Reasoning},
year={2024},
}

@InProceedings{collabotative,
  title = 	 {Collaborative Machine Learning with Incentive-Aware Model Rewards},
  author =       {Sim, Rachael Hwee Ling and Zhang, Yehong and Chan, Mun Choon and Low, Bryan Kian Hsiang},
  booktitle = 	 {Proc. Int. Conf. Mach. Learn. (ICML)},
  year = 	 {2020},
}

@article{anada2025measuring,
  title={Measuring Participant Contributions in Decentralized Federated Learning},
  author={Anada, Honoka and Kaneko, Tatsuya and Takamaeda-Yamazaki, Shinya},
  journal={arXiv preprint: 2505.23246  [cs.LG]},
  year={2025}
}

@inproceedings{gan_sv,
  author    = {Q. Li and X. Li and Z. Liu and H. Qi},
  title     = {pFedCE: Personalized Federated Learning Based on Contribution Evaluation},
  booktitle = {Proc. Int. Conf. Artif. Intell., Robotics, and Communication (ICAIRC)},
  year      = {2024},
  
  
}

@inproceedings{con_gan,
  author    = {A. Odena and C. Olah and J. Shlens},
  title     = {Conditional Image Synthesis with Auxiliary Classifier {GAN}s},
  booktitle = {Proc. Int. Conf. Mach. Learn. (ICML)},
  year      = {2017},
}

@article{privacy_gan_two,
  author  = {R. Venugopal and N. Shafqat and I. Venugopal and B. M. J. Tillbury and H. D. Stafford and A. Bourazeri},
  title   = {Privacy Preserving Generative Adversarial Networks to Model Electronic Health Records},
  journal = {Neural Networks},
  volume  = {153},
  pages   = {339--348},
  year    = {2022},
}

@article{privacy_gan_one,
  author  = {A. Yale and S. Dash and R. Dutta and I. Guyon and A. Pavao and K. P. Bennett},
  title   = {Generation and Evaluation of Privacy Preserving Synthetic Health Data},
  journal = {Neurocomputing},
  volume  = {416},
  pages   = {244--255},
  year    = {2020},
}

@inproceedings{Cifar10,
  author    = {K. He and X. Zhang and S. Ren and J. Sun},
  title     = {Deep residual learning for image recognition},
  booktitle = {Proc. IEEE Conf. Comput. Vision Pattern Recognition (CVPR)},
  year      = {2016},
}

@inproceedings{ISIC,
  author    = {N. C. F. Codella and D. Gutman and M. E. Celebi and B. Helba and M. A. Marchetti 
               and S. W. Dusza and A. Kalloo and K. Liopyris and N. Mishra 
               and H. Kittler and A. Halpern},
  title     = {Skin lesion analysis toward melanoma detection: A challenge at the 2017 {ISBI}, 
               hosted by the International Skin Imaging Collaboration ({ISIC})},
  booktitle = {Proc. IEEE Int. Symp. Biomedical Imaging (ISBI)},
  year      = {2018},
}

@article{Pcam,
  author  = {B. S. Veeling and J. Linmans and J. Winkens and T. Cohen and M. Welling},
  title   = {Rotation Equivariant {CNNs} for Digital Pathology},
  journal = {arXiv preprint: 1806.03962 [cs.CV]},
  month   = {Jun},
  year    = {2018}
}

@article{shamir79,
  author  = {A. Shamir},
  title   = {How to Share a Secret},
  journal = {Commun. ACM},
  volume  = {22},
  number  = {11},
  pages   = {612--613},
  month   = {Nov.},
  year    = {1979}
}

@article{xhemrishi2025detect,
  author  = {M. Xhemrishi and A. Graell i Amat and B. Pejó},
  title   = {Detect \& Score: Privacy-Preserving Misbehaviour Detection and Contribution Evaluation in Federated Learning},
  journal = {arXiv preprint: 2506.23583  [cs.CR]},  
  year    = {2025}
}

@INPROCEEDINGS{Leak,
  author={Wang, Zhibo and Song, Mengkai and Zhang, Zhifei and Song, Yang and Wang, Qian and Qi, Hairong},
  booktitle={IEEE Conference on Computer Communications (INFOCOM)}, 
  title={Beyond Inferring Class Representatives: User-Level Privacy Leakage From Federated Learning}, 
  year={2019},
}

@misc{brain,
  author       = {M. Nickparvar},
  title        = {Brain Tumor {MRI} Dataset},
  year         = {2021},
  howpublished = {\url{https://www.kaggle.com/datasets/masoudnickparvar/brain-tumor-mri-dataset}},
  note         = {Accessed: 2025-04-30}
}

@article{kairouz2021advances,
  author  = {P. Kairouz and H. B. McMahan and B. Avent and A. Bellet and M. Bennis and A. N. Bhagoji and et al.},
  title   = {Advances and open problems in federated learning},
  journal = {Found. Trends Mach. Learn.},
  volume  = {14},
  year    = {2021},
}

@book{cramer2015secure,
  author    = {R. Cramer and I. B. Damg{\aa}rd and J. B. Nielsen},
  title     = {Secure Multiparty Computation and Secret Sharing},
  publisher = {Cambridge University Press},
  year      = {2015}
}

@article{heyndrickx2023melloddy,
  title={Melloddy: Cross-pharma federated learning at unprecedented scale unlocks benefits in qsar without compromising proprietary information},
  author={Heyndrickx, Wouter and Mervin, Lewis and Morawietz, Tobias and Sturm, No{\'e} and Friedrich, Lukas and Zalewski, Adam and Pentina, Anastasia and Humbeck, Lina and Oldenhof, Martijn and Niwayama, Ritsuya and others},
  journal={Journal of chemical information and modeling},
  year={2023},
  publisher={ACS Publications}
}

@book{pejo2022guide,
  title={Guide to differential privacy modifications: a taxonomy of variants and extensions},
  author={Pej{\'o}, Bal{\'a}zs and Desfontaines, Damien},
  year={2022},
  publisher={Springer Nature}
}

@inproceedings{bonawitz2017practical,
  title={Practical secure aggregation for privacy-preserving machine learning},
  author={Bonawitz, Keith and Ivanov, Vladimir and Kreuter, Ben and Marcedone, Antonio and McMahan, H Brendan and Patel, Sarvar and Ramage, Daniel and Segal, Aaron and Seth, Karn},
  booktitle={Proc. ACM SIGSAC Conf. Computer and Communications Security},
  year={2017}
}

@InProceedings{data_banzhaf,
  title = 	 {Data Banzhaf: A Robust Data Valuation Framework for Machine Learning},
  author =       {Wang, Jiachen T. and Jia, Ruoxi},
  booktitle = 	 {Proceedings of The 26th International Conference on Artificial Intelligence and Statistics},
  year = 	 {2023},
  publisher =    {PMLR} 
}

@article{banzhaf,
  author  = {P. Dubey and L. S. Shapley},
  title   = {Mathematical properties of the {B}anzhaf power index},
  journal = {Math. Oper. Res.},
  volume  = {4},
  number  = {2},
  pages   = {99--131},
  year    = {1979}
}

@article{lin2024comprehensive,
  author  = {H. Lin and S. Wan and Z. Xie and K. Chen and M. Zhang and L. Shou and G. Chen},
  title   = {A Comprehensive Study of Shapley Value in Data Analytics},
  journal = {arXiv preprint: 2412.01460 [cs.DB]},
  year    = {2024}
}

@article{xhemrishi2023fedgt,
  author  = {M. Xhemrishi and J. Östman and A. Wachter-Zeh and A. Graell i Amat},
  title   = {{FedGT}: Identification of Malicious Clients in Federated Learning With Secure Aggregation},
  journal = {IEEE Trans. Inf. Forensics Security},
  volume  = {20},
  year    = {2025}
}

@inproceedings{black2021leave,
  author       = {E. Black and M. Fredrikson},
  title        = {Leave-one-out unfairness},
  booktitle    = {Proc. ACM Conf. Fairness, Accountability, and Transparency (FAccT)},      
  year         = {2021},
}

@article{evgeniou2004leave,
  author  = {T. Evgeniou and M. Pontil and A. Elisseeff},
  title   = {Leave one out error, stability, and generalization of voting combinations of classifiers},
  journal = {Mach. Learn.},
  volume  = {55},
  number  = {1},
  pages   = {1--30},
  year    = {2004}
}

@article{pejo2023quality,
  author  = {B. Pejó and G. Biczók},
  title   = {Quality Inference in Federated Learning with Secure Aggregation},
  journal = {IEEE Trans. Big Data},
  volume  = {9},
  number  = {5},
  year    = {2023}
}

@inproceedings{zhou2022survey,
  author    = {S. Zhou and M. Liao and B. Qiao and X. Yang},
  title     = {A Survey of Security Aggregation},
  booktitle = {Proc. Int. Conf. Adv. Commun. Technol. (ICACT)},
  year      = {2022}
}

@inproceedings{rozemberczki2022shapley,
  author    = {B. Rozemberczki and L. Watson and P. Bayer and H.-T. Yang and O. Kiss and S. Nilsson and et al.},
  title     = {The Shapley Value in Machine Learning},
  booktitle = {Proc. Int. Joint Conf. Artif. Intell. (IJCAI)},
  year      = {2022}
 
}

@InProceedings{ghorbani2019data,
  title = 	 {Data Shapley: Equitable Valuation of Data for Machine Learning},
  author =       {Ghorbani, Amirata and Zou, James},
  booktitle = 	 {Proc. Int. Conf. on Mach. Learn. (ICML)},
  year = 	 {2019},
}

@incollection{cohen2009pearson,
  author    = {I. Cohen and Y. Huang and J. Chen and J. Benesty},
  title     = {Pearson correlation coefficient},
  booktitle = {Noise Reduction in Speech Processing}, 
  publisher = {Springer},
  year      = {2009},
    
}

@incollection{abdi2007kendall,
  author    = {H. Abdi},
  title     = {The {K}endall rank correlation coefficient},
  booktitle = {Encyclopedia of Measurement and Statistics},     
  publisher = {Sage},
  address   = {Thousand Oaks, CA},
  year      = {2007},
           
}

@article{wu2021fast,
  author  = {H. Wu and P. Wang},
  title   = {Fast-Convergent Federated Learning with Adaptive Weighting},
  journal = {IEEE Trans. Cogn. Commun. Netw.},
  volume  = {7},
  number  = {4},
  year    = {2021}
}

@article{evans2024data,
  author  = {N. J. Evans and G. B. Mills and G. Wu and X. Song and S. McWeeney},
  title   = {Data valuation with gradient similarity},
  journal = {arXiv preprint: 2405.08217 [cs.LG]},
  year    = {2024}
}

@inproceedings{xu2021gradient,
  author  = {X. Xu and L. Lyu and X. Ma and C. Miao and C. S. Foo and B. K. H. Low},
  title   = {Gradient Driven Rewards to Guarantee Fairness in Collaborative Machine Learning},
  booktitle = {Proc. Adv. Neural Inf. Process. Syst. (NeurIPS)},
  year    = {2021}
}

@book{shapley1951notes,
  author    = {L. S. Shapley},
  title     = {Notes on the N-person Game},
  publisher = {Rand Corporation},
  year      = {1951}
}

@article{siomos2023contribution,
  author  = {V. Siomos and J. Passerat-Palmbach},
  title   = {Contribution Evaluation in Federated Learning: Examining Current Approaches},
  journal = {arXiv preprint: 2311.09856 [cs.LG]},
  year    = {2023}
}

@article{zar2005spearman,
  author  = {J. H. Zar},
  title   = {Spearman Rank Correlation},
  journal = {Encyclopedia of Biostatistics},
  year    = {2005}
}

@article{shapley1953value,
	title={A value for n-person games},
	author={Shapley, Lloyd S},
	journal={Contributions to the Theory of Games},
	year={1953}
}

@inproceedings{ma2021transparent,
  author    = {S. Ma and Y. Cao and L. Xiong},
  title     = {Transparent Contribution Evaluation for Secure Federated Learning on Blockchain},
  booktitle = {Proc. Int. Conf. Data Eng. Workshops (ICDEW)},
  year      = {2021}
}

@inproceedings{zheng2022secure,
    author = {S. Zheng and Y. Cao and M. Yoshikawa},
    title = {Secure Shapley Value for Cross-Silo Federated Learning},
    booktitle = {Proc. VLDB Endowment},
    year = {2023}
}

@article{watson2022differentially,
  author  = {L. Watson and R. Andreeva and H.-T. Yang and R. Sarkar},
  title   = {Differentially Private Shapley Values for Data Evaluation},
  journal = {arXiv preprint: 2206.00511 [cs.LG]},
  year    = {2022}
}

@article{kwon2021beta,
  title={Beta Shapley: a Unified and Noise-reduced Data Valuation Framework for Machine Learning},
  author={Kwon, Yongchan and Zou, James},
  journal={arXiv preprint arXiv:2110.14049},
  year={2021}
}

@article{liu2021gtg,
  author  = {Z. Liu and Y. Chen and H. Yu and Y. Liu and L. Cui},
  title   = {GTG-Shapley: Efficient and Accurate Participant Contribution Evaluation in Federated Learning},
  journal = {ACM Trans. Intell. Syst. Technol.},
  volume  = {13},
  number  = {4},
  pages   = {60:1--60:21},
  year    = {2022}
}

@inproceedings{song2019profit,
  author    = {T. Song and Y. Tong and S. Wei},
  title     = {Profit Allocation for Federated Learning},
  booktitle = {Proc. IEEE Int. Conf. Big Data (Big Data)},
  year      = {2019}
}

@inproceedings{zhu2019deep,
	title={Deep leakage from gradients},
	author={Zhu, Ligeng and Liu, Zhijian and Han, Song},
	booktitle={Advances in Neural Information Processing Systems},
	year={2019}
}

@inproceedings{ijcai2022p782,
  author    = {R. H. L. Sim and X. Xu and B. K. H. Low},
  title     = {Data Valuation in Machine Learning: "Ingredients", Strategies, and Open Challenges},
  booktitle = {Proc. Int. Joint Conf. Artif. Intell. (IJCAI)},
  year      = {2022}
}
